\documentclass[12pt,a4paper,amsmath,amssymb]{amsart}

\usepackage[latin1]{inputenc}
\usepackage{amsthm}
\usepackage{latexsym}
\usepackage{amsfonts}
\usepackage{bbm,bm,dsfont}   
\usepackage{color} 

\newtheorem{theorem}{Theorem}
\newtheorem{lemma}{Lemma}
\newtheorem{remark}{Remark}

\newtheorem{proposition}{Proposition}  
\newtheorem{example}{Example}

\newcommand{\N}{\mathbb N}
\newcommand{\R}{\mathbb R}

\newcommand{\C}{\mathbb C}

\newcommand{\hi}{\mathcal{H}} 
\newcommand{\hd}{{\mathcal{H}_\oplus}} 
\newcommand{\ki}{\mathcal{K}} 

\newcommand{\lh}{\mathcal{L(H)}} 
\newcommand{\lk}{\mathcal{L(K)}} 
\renewcommand{\th}{\mathcal{T(H)}} 
\newcommand{\tk}{\mathcal{T(K)}} 
\newcommand{\sh}{\mathcal{S(H)}} 

\newcommand{\tr}[1]{\mathrm{tr}\left[#1\right]} 
\newcommand{\kb}[2]{|#1\,\rangle\langle\,#2|} 

\newcommand{\Mo}{\mathsf{M}} 
\newcommand{\sfe}{\mathsf{E}} 
\newcommand{\M}{\mathcal M} 
\renewcommand{\O}{\mathrm{Obs}} 
\newcommand{\In}{\mathrm{Ins}} 
\newcommand{\ext}{\mathrm{Ext}\,} 

\newcommand{\hE}{\mathcal{E}} 

\def\<{\langle}
\def\>{\rangle}
\def\d{{\mathrm d}}
\newcommand{\bo}[1]{\mathcal{B}(#1)} 
\newcommand{\fii}{\varphi}
\newcommand{\la}{\lambda}
\newcommand{\tj}{\vartheta}
\newcommand{\CHI}[1]{\ensuremath{ \chi\raisebox{-1ex}{$\scriptstyle #1$} }}
\newcommand{\CHII}[1]{\ensuremath{{\hat\chi}\raisebox{-1ex}{$\scriptstyle #1$} }}
\newcommand{\ov}{\overline}

\newcommand{\lin}{{\rm lin}}

\newcommand{\id}{{\mathbbm1}}
\newcommand{\ve}{\bm{e}} 
\newcommand{\vn}{\bm{n}} 
\newcommand{\vp}{\bm{p}} 
\newcommand{\vsigma}{\bm{\sigma}} 

\newcommand{\mc}[1]{\mathcal{#1}}

\newcommand{\mb}[1]{\mathbb{#1}}

\begin{document} 

%
%
%
%
%
%


\title{Barycentric decomposition for quantum instruments}

\author{Juha-Pekka Pellonp\"a\"a}
\email{juhpello@utu.fi}
\address{Turku Centre for Quantum Physics, Department of Physics and Astronomy, University of Turku, FI-20014 Turku, Finland}
\author{Erkka Haapasalo}
\email{cqteth@nus.edu.sg}
\address{Centre for Quantum Technologies, National University of Singapore, Science Drive 2,
Block S15-03-18, Singapore 117543}
\author{Roope Uola}
\email{roope.uola@unige.ch}
\address{University of Geneva, 1211 Geneva, Switzerland}

\begin{abstract}
We present a barycentric decomposition for quantum instruments whose output space is finite-dimensional and input space is separable. As a special case, we obtain a barycentric decomposition for channels between such spaces and for normalized positive-operator-valued measures in separable Hilbert spaces. This extends the known results by Ali and Chiribella {\it et al.} on decompositions of quantum measurements, and formalises the fact that every instrument between finite-dimensional Hilbert spaces can be represented using only finite-outcome instruments.
\\

\end{abstract}

\maketitle

\section{Introduction}

Quantum measurement theory has been the center of considerable attention in the last years. Traditionally being a tool for quantum foundations, the precise and general description of the quantum measurement procedure has become important especially in the fields of quantum information and quantum technology. For example, the subfields of quantum correlations \cite{JMreview}, quantum state discrimination \cite{BaKw,Barnett}, and quantum thermodynamics \cite{GURU} bear deep connections to and even necessitate the use of the mathematical framework of generalised measurements.

By now, generalised measurements, mathematically described by normalized positive-operator-valued measures (POVMs), have become a standard description of quantum measurements in the field of quantum information theory. This, together with the foundational interest on the topic \cite{Measurement}, has spurred a plethora of structural and conceptual results on, e.g., quantification of measurement incompatibility \cite{Cava,Natu,PUSU,ONE}, informativeness of quantum measurements \cite{Nooa}, and the convex structure of the set of quantum measurements \cite{DAri,Parta,Pe11}. Of special interest to our work are extreme measurements which have been extensively further studied in \cite{HAHE,HAPE2,HAPEU,HEPO,Pell2,Ketsuppi}.

One can go a step deeper in the description of the quantum measurement procedure. Generalised measurements (POVMs) only consider the measurement statistics, but do not account for the state change due to the measurement. To include this into the picture, one uses quantum instruments \cite{Measurement}. These can be seen as generalisations of the projection postulate to the realm of generalised measurements. Quantum instruments have entered quantum information theory more recently than generalised measurements, but have already found various applications in, e.g., the subfields of temporal correlations \cite{AM,macro}, quantum memories \cite{Cope,Ion,Seka}, self-testing \cite{Moh}, quantum thermodynamics \cite{GURU,BUL},  and programmability of quantum measurement devices \cite{Lin}. In contrast to generalised measurements, the structural results on quantum instruments motivated by the applications of quantum information have started to appear only very recently. To mention a few, resource-theoretic approach to properties of instruments was considered in \cite{Haapasalo2015,UoBu,UKA}, implementation of joint measurements in \cite{HAPE,HeMi2013,retriivinki}, and the geometric structure of instruments in \cite{DAPeSe,HaPe2021,instru1}. Convex structure of unconditioned state transformations, i.e.,\ quantum channels has also been studied extensively \cite{Tsuikkis}.

In this paper, we contribute to the geometric structure of quantum instruments by presenting a barycentric decomposition of quantum instruments. Our results work for instruments with a finite-dimensional output space and a separable input (Hilbert) space. This includes as a special case quantum channels between the mentioned spaces and generalised measurements acting on a separable Hilbert space. The latter generalises the well-known finite-dimensional result \cite{ChDASc}. Our main result also extends the known result stating that every finite-dimensional generalised measurement can be represented using finite-outcome measurements into the realm of finite-dimensional quantum instruments, which can be of interest to optimisation problems where all instruments have to be considered.

To our knowledge, the investigation into the convex decompositions of quantum measurements into extremes was begun by Ali in \cite{Ali} (see also \cite{Bene}). In fact, this work discussed sub-normalized positive-operator-valued measures in von Neumann algebras. However, in this work, POVMs could be given as barycentres of probability measures over the extreme points of the convex set of sub-normalized positive-operator-valued measures which, in the quantum case, also include measures which cannot be considered as generalized measurements. In \cite{ChDASc2}, Chiribella {\it et al.} study POVMs in finite-dimensional Hilbert spaces with standard Borel value spaces, finding barycentre presentations for them over the set of extreme generalized measurements. Moreover, in this finite-dimensional setting, extreme measurements are supported on a finite set \cite{ChDASc}. An important ingredient in this work is Lemma 6 of \cite{ChDASc2} which guarantees that the probability measure, whose barycentre the quantum measurement is, is actually supported by the set of extreme measurements. We will apply similar techniques in our work where we generalize this to a separable Hilbert space.

We start by reviewing the basic machinery we need for generalised measurements (POVMs) and instruments in Section \ref{sec:measurext}. We also recall the extremality results for these objects. In Section \ref{sec:bary} we go on to apply Choquet theory to our instruments; as good reference material for this, we mention \cite{Alfsen}. This section culminates with expressing instruments with separable input Hilbert space and finite-dimensional output space as barycentres over probability measures supported by the set of extreme instruments. As a special case of this, we obtain the infinite-dimensional generalization of the related result of \cite{ChDASc2}. Finally, we give barycentric decompositions for measurements, channels, and states obtained as special cases of our main result. We also consider an example involving qubit measurements. However, we start by considering a simple example of spin direction measurements.

\subsection{Example: Spin direction}

Let $\mathbb S^2:=\big\{\vn=(n^1,n^2,n^3)\in\R^3\,\big|\,(n^1)^2+(n^2)^2+(n^3)^2=1\big\}$ be the unit sphere equipped with the Borel $\sigma$-algebra $\mathcal B(\mathbb S^2)$ of the usual topology.
We will use the spherical coordinate parametrization $n^1=\sin\theta \cos\fii$, $n^2=\sin\theta \sin\fii$,
$n^3=\cos\theta$, $\theta\in[0,\pi]$, $\fii\in[0,2\pi)$, and the area measure $\d\vn=\sin\theta\,\d\theta\,\d\fii$ giving
$\int_{\mathbb S^2}\d\vn=4\pi$. The (qubit) spin direction observable \cite[p.\ 110]{Teikonkirja} is
$$
\mathsf D(X):=\frac{1}{4\pi}\int_X \big(\id+\vn\cdot\vsigma\big)\d\vn,\qquad X\in\mathcal B(\mathbb S^2);
$$
here 
$$
\sigma_0=\id=\left(\begin{matrix} 1&0\\ 0&1 \end{matrix}\right),\quad
\sigma_1=\left(\begin{matrix} 0&1\\ 1&0 \end{matrix}\right),\quad
\sigma_2=\left(\begin{matrix} 0&-i\\ i&0 \end{matrix}\right),\quad
\sigma_3=\left(\begin{matrix} 1&0\\ 0&-1 \end{matrix}\right)
$$
are the Pauli matrices and
$$
P_{\vn}:= \frac12(\id+\vn\cdot\vsigma)
=\frac12\begin{pmatrix}
1+n^3 & n^1-i n^2 \\
n^1+i n^2 & 1-n^3
\end{pmatrix}
=\frac12\begin{pmatrix}
1+\cos\theta & \sin\theta \, e^{-i\fii} \\
\sin\theta \, e^{i\fii} & 1-\cos\theta
\end{pmatrix}
$$
is a rank-1 projection for all $\vn\in\mathbb S^2$. Since
$$
\int_{\mathbb S^2} \cos(2\fii)\d\mathsf D(\vn)=0
$$
it follows that $\mathsf D$ is not an extreme element of the convex set of qubit POVMs on $\mathcal B(\mathbb S^2)$ \cite{Pe11}. Indeed, we can write $\mathsf D=\frac12\mathsf D_++\frac12\mathsf D_-$
where $\mathsf D_{\pm}(X):=\int_{X} [1\pm \cos(2\fii)]\d\mathsf D(\vn)$ are unequal POVMs.
Since extreme POVMs are discrete in finite dimensions \cite{ChDASc} one cannot write $\mathsf D$ as a finite (or countably infinite) convex combination of extreme POVMs.
However, if we define projection valued measures\footnote{Projection valued measures are automatically extreme \cite{Pe11}.} 
$$
\mathsf S_{\vn}(X):=P_{\vn}\delta_{\vn}(X)+P_{-\vn}\delta_{-\vn}(X),\qquad
X\in\mathcal B(\mathbb S^2),
$$
where $\delta_{\vn}$ is a Dirac (point) measure at $\vn$, we can write
$$
\mathsf D(X)=\frac{1}{2\pi}\int_{\mathbb S^2_+} \mathsf S_{\vn}(X)\d\vn,\qquad X\in\mathcal B(\mathbb S^2),
$$
where $\mathbb S^2_+:=\big\{(n^1,n^2,n^3)\in\mathbb S^2\,\big|\,n_3\ge 0\big\}$ with the area $2\pi$. 
Since now $(2\pi)^{-1}\d\vn=:\d p_{\mathsf D}(\mathsf S_{\vn})$ can trivially be extended to the probability measure $p_{\mathsf D}$ on the set of all extreme qubit POVMs on $\mathcal B(\mathbb S^2)$, we see that $\mathsf D$ is the {\it barycenter} of $p_{\mathsf D}$ (or a `continuous convex combination').
 
\section{Sets of measurements and instruments and their extreme points}\label{sec:measurext}

For any Hilbert space $\hi$ we let $\lh$ [resp.\ $\th$] denote the Banach space of bounded [resp.\ trace-class] operators on $\hi$. The operator norm of $\lh$ is denoted by $\|\,\cdot\,\|$ and the trace norm of $\th$ by $\|\,\cdot\,\|_1$. We say that a positive operator $\rho\in\th$ of trace 1 is a state (or a density operator) and denote the convex set of states by $\sh$. The identity operator of any Hilbert space $\hi$ is denoted by $\id_\hi$. Throughout this article, we let $\hi$ and $\ki$ be {\it separable} (complex) nontrivial Hilbert spaces and $(\Omega,\Sigma)$ be a measurable space (i.e.\ $\Sigma$ is a $\sigma$-algebra of subsets of a nonempty set $\Omega$). For any ($\sigma$-finite) measure $\mu:\,\Sigma\to[0,\infty]$ and for all $1\le p\le\infty$, we let $L^p(\mu)$ denote the corresponding Lebesgue space, and let $\hd$ denote a direct integral $\int_\Omega^\oplus\hi_x\d\mu(x)$ of {separable} Hilbert spaces $\hi_x$. For each $f\in L^\infty(\mu)$, we denote briefly by $\hat f$ the multiplicative (i.e.\ diagonalizable) bounded operator $(\hat f\psi)(x):=f(x)\psi(x)$ for all $\psi\in\hd$ and $\mu$--almost all $x\in\Omega$. Especially, one has the {canonical spectral measure} $\Sigma\ni X\mapsto\CHII X\in\mathcal L(\hd)$ (where $\CHI X$ is the characteristic function of  $X$).

If $\Omega$ is a topological space, we always let $\Sigma$ be its Borel $\sigma$-algebra $\mathcal B(\Omega)$ and denote by $C(\Omega)$ the set of continuous functions $\Omega\to\C$. Let $\Omega$ be a locally compact Hausdorff space, and let $C_0(\Omega)\subseteq C(\Omega)$ consist of (bounded) functions vanishing at infinity;\footnote{That is, if $f\in C_0(\Omega)$ then, for each $\epsilon>0$, there exists a compact set $K$ such that $\sup_{x\in\Omega\setminus K}|f(x)|<\epsilon$. Now $f$ can uniquely be extended to $C(\ov\Omega)$ by setting $f(\infty):=0$ and this embedding $C_0(\Omega)\to C(\ov\Omega)$ is an isometry with closed image.} especially, if $\Omega$ is compact then $C_0(\Omega)= C(\Omega)$.  As usual, we equip $C_0(\Omega)$ with the norm $f\mapsto\|f\|_\infty:=\sup_{x\in\Omega}|f(x)|$ so $C_0(\Omega)$ becomes a commutative $C^*$--algebra. If $\Omega$ is not compact we let $\ov\Omega:=\Omega\biguplus\{\infty\}$ be the Alexandroff (one-point) compactification of $\Omega$; it is a compact Hausdorff space with an extra (unphysical) point $\infty$, the point at infinity.  [If $\Omega$ is already compact we denote $\ov\Omega=\Omega$, i.e.\ no compactification is needed.] Now $C(\ov\Omega)$ can be interpreted as a unitalization of $C_0(\Omega)$ as follows: For any $g\in C(\ov\Omega)$ write\footnote{That is, $g$ is `constant at infinity.'}
$$
g(x)=\underbrace{g(x)-g(\infty)}_{=:\,f(x)}+\underbrace{g(\infty)}_{\in\C},\qquad x\in\ov\Omega, 
$$
where the restriction $f|_\Omega\in C_0(\Omega)$ and $\CHI{\ov\Omega}$ is the unit, so one has 
$C(\ov\Omega)\cong C_0(\Omega)\oplus\C$. Finally,  we note that $C_0(\Omega)$ [and $C(\ov\Omega)$] is separable if and only if $\Omega$ is second countable.

The Riesz-Markov-Kakutani representation theorem says that for any bounded (i.e.\ continuous) linear map $F:\,C_0(\Omega)\to\C$ there is a {\it unique regular} (Radon) complex measure $\mu_F$ on $\Omega$ such that $F(f)=\int_\Omega  f(x)\d\mu_F(x)$ for all $f\in C_0(\Omega)$.\footnote{The norm $\sup\{|F(f)|\;|\;\|f\|_\infty\le1\}$ of $F$ equals the total variation norm $|\mu_F|(\Omega)$ of $\mu_F$, and $F$ is positive if and only if $\mu_F$ is positive.} Hence, the topological dual of $C_0(\Omega)$ is the Banach space $\mathrm{rca}(\Omega)$ of the regular measures $\mathcal B(\Omega)\to\C$. We equip $\mathrm{rca}(\Omega)$ with the weak${}^*$ topology which means that the maps $\mu\mapsto\int_\Omega f(x)\d\mu(x)$ are continuous and a net $(\mu_i)_{i\in\mathcal I}\subseteq\mathrm{rca}(\Omega)$ (where $\mathcal I$ is a directed set) converges to $\mu\in\mathrm{rca}(\Omega)$ if $\lim_{i\in\mathcal I}\int_\Omega f(x)\d\mu_i(x)=\int_\Omega f(x)\d\mu(x)$ for all $f\in C_0(\Omega)$. By the Banach-Alao\u{g}lu theorem, for each $r>0$, the `closed ball'  ${\rm ball}(\Omega)_r:=\big\{\mu\in\mathrm{rca}(\Omega)\;\big|\;|\mu|(\Omega)\le r\big\}$ is weak${}^*$ compact.

Note that any $\mu\in\mathrm{rca}(\Omega)$ can be extended to $\mathcal B(\ov\Omega)$ by setting $\mu(\{\infty\}):=0$. Hence, $\mathrm{rca}(\Omega)\subseteq\mathrm{rca}(\ov\Omega)$. On the other hand, for each $\nu\in\mathrm{rca}(\ov\Omega)$, one can split
$$
\nu=\underbrace{[\nu-\nu(\{\infty\})\delta_\infty]}_{=:\,\mu}+\underbrace{\nu(\{\infty\})}_{\in\C}\delta_\infty
$$
where $\mu(\{\infty\})=0$ and $\mu|_{\mathcal B(\Omega)}=\nu|_{\mathcal B(\Omega)}\in\mathrm{rca}(\Omega)$; here $\delta_\infty$ is the Dirac point measure at $\infty$. Thus, $\mathrm{rca}(\ov\Omega)\cong\mathrm{rca}(\Omega)\oplus\C$.\footnote{But usually $\int_{\ov\Omega}g(x)\d\nu(x)=\int_{\Omega}f(x)\d\mu(x)+g(\infty)\nu(\ov\Omega)\ne\int_{\Omega}f(x)\d\mu(x)+g(\infty)\nu(\{\infty\})$.}

If $\Omega$ is a locally compact second countable Hausdorff space, then $\Omega$ and $\ov\Omega$ are Polish  (i.e.\ separable and completely metrizable) spaces and any complex measure is regular so, e.g., $\mathrm{rca}(\Omega)=\mathrm{ca}(\Omega)$, the space of all complex measures on $\Omega$; in addition, $\big(\Omega,\mathcal B(\Omega)\big)$ is a standard Borel (measurable) space so that $\Omega$ is countable\footnote{That is, finite or countably infinite.} (and $\mathcal B(\Omega)=2^\Omega$) or $\big(\Omega,\mathcal B(\Omega)\big)$ is isomorphic to $\big(\R,\mathcal B(\R)\big)$.\footnote{Usually, in physics, this holds in `continuous cases' where $\Omega$ is a finite-dimensional second countable Hausdorff manifold (which is automatically locally compact).} Since now $C_0(\Omega)$ is separable, i.e., has a dense sequence $\{f_n\}_{n=1}^\infty$, the weak${}^*$ topology of ${\rm ball}_r(\Omega)$ is metrizable\footnote{The whole $\mathrm{ca}(\Omega)$ is not necessarily metrizable.} with a metric 
$$
d(\mu,\nu):=\sum_{n=1}^\infty \frac{\big|\int f_n\d\mu-\int f_n\d\nu\big|}{2^n\Big(1+\big|\int f_n\d\mu-\int f_n\d\nu\big|\Big)}.
$$

\begin{remark}\rm\label{Remark1}
The convex set of probability\footnote{A measure $\mu\in\mathrm{rca}(\Omega)$ belongs to ${\rm prob}(\Omega)$ iff $\mu(X)\ge 0$ for all $X\in\bo\Omega$ [or $\int_\Omega f(x)\d\mu(x)\ge0$ for all $f\in C_0(\Omega)$ such that $f\ge0$] and $\mu(\Omega)=1$; this normalization condition cannot be expressed as an integral of a function of $C_0(\Omega)$ if $\Omega$ is not compact (since then $\CHI\Omega\notin C_0(\Omega)$). However, $\CHI{\ov\Omega}\in C(\ov\Omega)$ so $\lim_{i\in\mathcal I}\int_{\ov\Omega} \CHI{\ov\Omega}(x)\d\mu_i(x)=1$ and $\lim_{i\in\mathcal I}\int_{\ov\Omega} g(x)\d\mu_i(x)\ge0$ for all $g\ge0$ so that ${\rm prob}(\ov\Omega)$ is closed and compact [here $(\mu_i)_{i\in\mathcal I}\subseteq{\rm prob}(\ov\Omega)$ is assumed to converge in $\mathrm{rca}(\ov\Omega)$; actually in ${\rm ball}_1(\ov\Omega)$].} (Radon) measures ${\rm prob}(\ov\Omega)\subset\mathrm{rca}(\ov\Omega)$ is weak${}^*$ compact but ${\rm prob}(\Omega)$ is not (necessarily) compact: The basic example is the sequence of Dirac (point) measures $n\mapsto\delta_n$ on $\R$ which converges to the zero measure since $\lim_{n\to\infty}\int_\R f(x)\d\delta_n(x)=\lim_{n\to\infty} f(n)=0$ for all $f\in C_0(\R)$. However, on the (compact) extended  real line $\ov\R=\R\cup\{\infty\}$ one gets $\lim_{n\to\infty}\int_{\ov\R} g(x)\d\delta_n(x)=\lim_{n\to\infty} g(n)=g(\infty)$ for all $g\in C(\ov\R)$, that is, $\delta_n\to\delta_\infty$.\footnote{This is the main reason why we use the compactification $\ov\Omega$ instead of $\Omega$ in the following.}
\end{remark}

\subsection*{Operator measures}

Let $\Mo:\,\Sigma\to\lh $ be an {\it operator (valued) measure,} i.e.\ (ultra)weakly $\sigma$-additive mapping. We call $\Mo$ \emph{positive} if for all $X\in\Sigma$,  $\Mo(X)\geq 0$, \emph{normalized} if $\Mo(\Omega)=\id_\hi$, and \emph{projection valued} if $\Mo(X)^2=\Mo(X)^*=\Mo(X)$ for all $X\in\Sigma$.
Normalized positive operator valued measures (POVMs) or {\it semispectral measures} are identified with \emph{(quantum) observables} whereas normalized projection valued measures (PVMs) are called \emph{spectral measures} or \emph{sharp observables}. The number $\tr{\rho\Mo(X)}\ge 0$ is interpreted as the probability of getting a measurement outcome $x$ belonging to the set $X\in\Sigma$ when the system is in the state $\rho\in\sh$ and a measurement of the POVM $\Mo:\,\Sigma\to\lh$ is performed. The convex set of POVMs $\Mo:\Sigma\to\lh $ is denoted by $\O(\Sigma,\,\hi)$ and its extreme points by $\ext\O(\Sigma,\,\hi)$. A convex combination (observable) $t\Mo_1+(1-t)\Mo_2$, $0<t<1$, can be viewed as a randomization of measuring procedures represented by the observables $\Mo_1$ and $\Mo_2$. An extreme observable $\Mo\in\ext\O(\Sigma,\,\hi)$ cannot be obtained as a (nontrivial) convex combination meaning that the measurement of $\Mo$ involves no redundancy caused by mixing different measuring schemes.

\subsection*{Instruments}

Let $\hi$ and $\ki$ be Hilbert spaces. A linear map $\Phi:\mc L(\ki)\to\mc L(\hi)$ is an {\it operation} if it is {\it completely positive} (CP) and $\Phi(\id_\ki)\leq\id_\hi$. If $\Phi(\id_\ki)=\id_\hi$ (unitality), $\Phi$ is a {\it channel}. If the operation $\Phi$ is normal, i.e.,\ for any increasing net $(B_i)_{i\in\mathcal I}$ of selfadjoint operators $B_i\in\mc L(\ki)$ bounded from above we have $\Phi\big(\sup_{i\in\mathcal I}B_i)=\sup_{i\in\mathcal I}\Phi(B_i)$ (or equivalently $\Phi$ is ultraweakly continuous), we say that $\Phi$ is a {\it quantum operation} (or {\it quantum channel} if $\Phi$ is unital). We say that a map $\M:\,\Sigma\times \lk\to \lh$ is a {\it (Heisenberg) instrument} if
\begin{itemize}
\item[(i)] for all $X\in\Sigma$, the mapping $\lk\ni B\mapsto \M(X,B)\in\lh$ is an operation,
\item[(ii)] $\M(\Omega,\id_\ki)=\id_\hi$, and
\item[(iii)] $\tr{\rho\M(\cup_{n=1}^\infty X_n,B)}=\sum_{n=1}^\infty\tr{\rho\M(X_n,B)}$ for any pairwise disjoint sequence $\{X_n\}_{n=1}^\infty\subseteq\Sigma$ and for all $\rho\in\th$, $B\in\lk$. 
\end{itemize}
If, additionally,
\begin{itemize}
\item[(iv)] for all $X\in\Sigma$, $B\mapsto\mc M(X,B)$ is normal (i.e.,\ a quantum operation),
\end{itemize}
we say that $\mc M$ is a {\it quantum instrument}. Let $\M$ be a Heisenberg instrument. For any $B\in\lk$, we define  an operator measure
$$
\Mo^B:\,\Sigma\to\lh,\,X\mapsto \Mo^B(X):=\M(X,B),
$$ 
see item (iii) above. It is positive if $B\ge 0$ and normalized if $B=\id_\ki$. Hence, $\Mo^{\id_\ki}$ is a POVM, the {\it associate observable of $\M$.} Also, $\M$ defines a channel $B\mapsto\M(\Omega,B)$, the {\it associate channel of $\M$} which is a quantum channel if $\M$ is a quantum instrument. Moreover, if $\Phi:\,\lk\to\lh$ is a [quantum] channel, then by choosing $\Omega=\{0\}$ and $\Sigma=2^{\{0\}}=\big\{\emptyset,\{0\}\big\}$, one can define a  [quantum] instrument $\M_\Phi(\{0\},B):=\Phi(B)$, $B\in\lk$. Similarly, for any POVM $\Mo:\,\Sigma\to\lh$ there exist a  quantum instrument $\M^\Mo:\,\Sigma\times\mathcal \C\to\lh$ defined by $\M^\Mo(X,c):=c\, \Mo(X)$ where $c\in\C\cong\mathcal L(\C)$ (via $c\mapsto c\kb11$). We call the instruments $\M_\Phi$ and $\M^\Mo$ {\it trivial instruments} associated with $\Phi$ and $\Mo$, respectively. Thus, it follows that all general results for instruments are applicable to channels and POVMs.

Using the duality $\mc T(\hi)^*=\mc L(\hi)$, any quantum operation $\Phi$ can be identified with a trace-non-increasing completely positive linear map $\Phi_*:\mc T(\hi)\to\mc T(\ki)$ through
$$
\tr{\Phi_*(\rho)B}=\tr{\rho\Phi(B)}
$$
for all $\rho\in\mc T(\hi)$ and $B\in\mc L(\ki)$. The map $\Phi\mapsto\Phi_*$ is a bijection from the set of quantum operations onto the set of trace-non-increasing completely positive\footnote{The product maps $\Lambda_n:\mc T(\hi\otimes\C^n)\to\mc T(\ki\otimes\C^n)$ are positive for all $n\in\N$ where $\Lambda_n(\rho\otimes A)=\Lambda(\rho)\otimes A$ for all $\rho\in\mc T(\hi)$ and all $(n\times n)$-complex matrices $A$.} linear maps $\Lambda:\,\mc T(\hi)\to\mc T(\ki)$.  We call the maps $\Lambda$ of the latter set as {\it Schr\"{o}dinger operations}. Thus, a quantum operation $\Phi$ can be expressed equivalently in the Heisenberg picture $\Phi:\mc L(\ki)\to\mc L(\hi)$ or in the Schr\"{o}dinger picture $\Phi_*:\mc T(\hi)\to\mc T(\ki)$. Motivated by this, we call a map $\mc I:\Sigma\times\mc T(\hi)\to\mc T(\ki)$ a {\it Schr\"{o}dinger instrument} if
\begin{itemize}
\item[(i)] for all $X\in\Sigma$, the mapping $\mc T(\hi)\ni \rho\mapsto \mc I(X,\rho)\in\mc T(\ki)$ is a Schr\"{o}dinger operation,
\item[(ii)] $\mc T(\hi)\ni \rho\mapsto\mc I(\Omega,\rho)\in\mc T(\ki)$ is trace-preserving, and
\item[(iii)] $\tr{\mc I(\cup_{n=1}^\infty X_n,\rho)B}=\sum_{n=1}^\infty\tr{\mc I(X_n,\rho)B}$ for any pairwise disjoint sequence $\{X_n\}_{n=1}^\infty\subseteq\Sigma$ and for all $\rho\in\th$, $B\in\lk$.
\end{itemize}
Now the map $\M\mapsto\M_*$ defined through
$$
\tr{\M_*(X,\rho)B}=\tr{\rho\M(X,B)}
$$
for all $X\in\Sigma$, $\rho\in\th$, and $B\in\lk$ is a bijection of the set of quantum instruments onto the set of Schr\"{o}dinger instruments. When $\mc I=\mc M_*$ for a  quantum instrument $\mc M$, we also denote $\mc M=\mc I^*$.

The convex set of quantum instruments $\M:\,\Sigma\times \lk\to \lh$ is denoted by $\In(\Sigma,\,\ki,\,\hi)$ and its extreme points by $\ext\In(\Sigma,\,\ki,\,\hi)$. When we drop normality, we denote the set of Heisenberg instruments by $\In_{\rm H}(\Sigma,\,\ki,\,\hi)$ and the set of the extreme points of this convex set by $\ext\In_{\rm H}(\Sigma,\,\ki,\,\hi)$. We may give the following extremality characterization for instruments. Note that the measure $\mu$ of the theorem below exists; one can choose $\mu(X)=\tr{\rho\M(X,\id_\ki)}$ for all $X\in\Sigma$ where $\rho$ is a faithful state on $\hi$, i.e.,\ all the eigenvalues of $\rho\in\sh$ are strictly positive.

\begin{theorem}\label{th2}\cite{instru1}
Let $\M\in\In(\Sigma,\,\ki,\,\hi)$ be a quantum instrument and $\mu:\,\Sigma\to[0,\infty]$ a $\sigma$-finite measure such that $\M(\cdot,\id_{\mc K})$ is absolutely continuous with respect to $\mu$. There exists a direct integral $\hd=\int_\Omega^\oplus\hi_x\d\mu(x)$ (with $\dim\hi_x\le\dim\hi\dim\ki$) such that, for all $X\in\Sigma$ and $B\in\lk$,
\begin{enumerate}
\item $\M(X,B)=Y^*(B\otimes\CHII X)Y$ where $Y:\,\hi\to\ki\otimes\hd$
is an isometry such that 
$$
\lin\big\{(B'\otimes\CHII{X'})Y\psi\,\big|\,B'\in\lk,\;X'\in\Sigma,\;\psi\in\hi\big\}
$$
is dense in $\ki\otimes\hd$ (a minimal Stinespring dilation for $\M$).  
\item $\M\in\ext\In(\Sigma,\ki,\hi)$ if and only if, for any decomposable operator $D=\int_\Omega^\oplus D(x)\d\mu(x)\in\mathcal L(\hd)$, the condition $Y^*(\id_{\ki}\otimes D)Y=0$ implies $D=0$.
\item Let $\Omega$ be a second countable Hausdorff space, $\dim\hi <\infty$, and $\M\in\ext\In(\bo\Omega,\ki,\hi)$.
Then $\M$ is concentrated on a finite set, i.e.,
\begin{equation*}
\M(X,B)=\sum_{i=1}^N \CHI X(x_i) \Phi_i(B),\qquad X\in\bo\Omega,\;B\in\lk,
\end{equation*}
for some finite number $N\leq(\dim\hi)^2$ of elements $x_1,\ldots,x_N\in\Omega$ and operations $\Phi_i:\,\lk\to\lh$.
\end{enumerate}
\end{theorem}

\section{Barycentric decompositions for instruments}\label{sec:bary}

From now on {\it we assume that $\Omega$ is a locally compact second countable Hausdorff space}, fix an 
$\M\in\In\big(\bo\Omega,\,\ki,\,\hi\big)$, and let $\mu$, $\hd$, and $Y$ be as in Theorem \ref{th2}. We let $\Mo^B=\M(\cdot,B)$ be the associated operator measures for all $B\in\lk$. Now $\hd$ is separable and, for all $f\in C_0(\Omega)$, $B\in\lk$, the operator integral
$$
L(f,\Mo^B):=\int_\Omega f(x)\d\Mo^B(x)=Y^*(B\otimes\hat f)Y\in\lh
$$
with the norm $\|L(f,\Mo^B)\|\le \|f\|_\infty\|B\|$. By denoting $\mu^{\M}_{\rho,B}(X):=\tr{\rho\Mo^B(X)}$, we get
$$
\tr{\rho L(f,\Mo^B)}=\tr{\rho Y^*(B\otimes\hat f)Y}=\int_\Omega f(x)\d \mu^{\M}_{\rho,B}(x)
$$
and $|\tr{\rho L(f,\Mo^B)}|\le\|\rho\|_1\|f\|_\infty\|B\|$ for all $\rho\in\th$. Hence, $|\mu^{\M}_{\rho,B}|(\Omega)\le \|\rho\|_1\|B\|$ and
$$
\th\times\lk\ni(\rho,B)\mapsto \mu^{\M}_{\rho,B}\in \mathrm{ca}(\Omega)
$$
is a bounded bilinear map; we say that a bilinear map $Q:\,\th\times\lk\to\mathrm{ca}(\Omega)$ is bounded if its norm 
\begin{eqnarray*}
\|Q\|_{\rm bil}&:=&\sup\big\{|Q(\rho,B)|(\Omega)\,\big|\,\|\rho\|_1\le 1,\;\|B\|\le 1\big\}\\
&=&
\sup\bigg\{\Big|\int_\Omega f(x) \d[Q(\rho,B)](x)\Big|\,\bigg|\,\|f\|_\infty\le1,\;\|\rho\|_1\le 1,\;\|B\|\le 1\bigg\}
\end{eqnarray*}
is finite. Note that $|Q(\rho,B)|(\Omega)\le\|Q\|_{\rm bil}\|\rho\|_1\|B\|$ and
$$
\Big|\int_\Omega f(x) \d[Q(\rho,B)](x)\Big|\le\|Q\|_{\rm bil}\|\rho\|_1\|f\|_\infty\|B\|
$$
for all $f\in C_0(\Omega)$, $\rho\in\th$, and $B\in\lk$. Clearly, such bounded bilinear maps form a Banach space \cite[E 2.3.9., p.\ 61]{Pedersen} whose unit ball\footnote{That is, the norm of its element $Q$ is $\|Q\|_{\rm bil}\le1$.} is denoted by $\rm{Ball}(\hi,\,\ki,\,\Omega)$. We may consider $\In\big(\bo\Omega,\,\ki,\,\hi\big)$ as a convex subset of $\rm{Ball}(\hi,\,\ki,\,\Omega)$. Moreover, any $Q\in\rm{Ball}(\hi,\,\ki,\,\Omega)$ can be viewed as an element of $\In_{\rm H}\big(\bo\Omega,\,\ki,\,\hi\big)$ if it satisfies the following conditions:
\begin{enumerate}

\item[(CP)] $\sum_{s,t=1}^n Q(\kb{\psi_t}{\psi_s},B_s^*B_t)\ge 0$ (i.e.\ the sum is a positive measure) for all $\psi_s\in\hi$, $B_s\in\lk$, and $s=1,2,\ldots,n\in\N$ (the CP-condition \cite{PeYl}),

\item[(NO)] $Q(\rho,\id_\ki)\in{\rm prob}(\Omega)$ for all $\rho\in\sh$ (normalization).\footnote{Cleary, the trace-preserving condition or unitality condition (2) follows from the fact that any trace-class operator can be written as a linear combination of four states (which are the normalized positive and negative parts of the real and imaginary parts).}

\end{enumerate}
Clearly now $Q$ determines a Heisenberg instrument $\M$ (i.e.\ the related linear map $B\mapsto\Mo^B$) via the formula 
$$
\tr{\rho L(f,\Mo^B)}=\int_\Omega f(x) \d[Q(\rho,B)](x),
$$
that is, $\mu^{\M}_{\rho,B}=Q(\rho,B)$, and we have $\|Q\|_{\rm bil}=1$ since $\mu^{\M}_{\rho,{\id_\ki}}(\Omega)=1$ for any state $\rho$. 

Recall that any 
$
{\rm ball}_ {\|\rho\|_1\|B\|}(\Omega)=\big\{\mu\in\mathrm{ca}(\Omega)\;\big|\;|\mu|(\Omega)\le  \|\rho\|_1\|B\|\big\}
 $ 
is weak${}^*$ compact and equip the cartesian product
\begin{equation}
\label{cart}
\prod_{(\rho,B)\in\th\times\lk}{\rm ball}_ {\|\rho\|_1\|B\|}(\Omega)
\end{equation}
with the product topology, that is, the coarsest topology for which all projections $Q\mapsto Q(\rho,B)$ are continuous. Hence, its net $(Q_i)_{i\in\mathcal I}$ converges to $Q$ if $\lim_{i\in\mathcal I}Q_i(\rho,B)=Q(\rho,B)$ for all $\rho\in\th$ and $B\in\lk$, that is, if
$$
\lim_{i\in\mathcal I}\int_\Omega f(x)\d[Q_i(\rho,B)](x)=\int_\Omega f(x)\d[Q(\rho,B)](x)
$$
for all $f\in C_0(\Omega)$, $\rho\in\th$ and $B\in\lk$. It follows from Tychonoff's theorem \cite{Pedersen} that the cartesian product is compact and, hence, its subset $\rm{Ball}(\hi,\,\ki,\,\Omega)$ is compact since it is closed: Indeed, if $(Q_i)_{i\in\mathcal I}\subseteq \rm{Ball}(\hi,\,\ki,\,\Omega)$ converges to $Q$ then, for all $c\in\C$, $\rho,\rho'\in\th$, and $B\in\lk$, one gets $Q(c\rho+\rho',B)=\lim_{i\in\mathcal I}Q_i(c\rho+\rho',B)=\lim_{i\in\mathcal I}[c Q_i(\rho,B)+Q_i(\rho',B)]=c \lim_{i\in\mathcal I} Q_i(\rho,B)+\lim_{i\in\mathcal I} Q_i(\rho',B)=c Q(\rho,B)+Q(\rho',B)
$ and similarly for the second argument showing bilinearity of $Q$. Automatically, $Q(\rho,B)\in{\rm ball}_ {\|\rho\|_1\|B\|}(\Omega)$, i.e.\ $|Q(\rho,B)|(\Omega)\le  \|\rho\|_1\|B\|$. By the just shown bilinearity, $Q$ has the norm  $\|Q\|_{\rm bil}\le1$, i.e.\ $Q\in\rm{Ball}(\hi,\,\ki,\,\Omega)$. To conclude, if we equip the vector space of bounded bilinear maps $Q:\,\th\times\lk\to\mathrm{ca}(\Omega)$ with the locally convex Hausdorff\footnote{Clearly, if $Q\ne0$ then there are $\rho$ and $B$ such that $Q(\rho,B)$ is not a zero measure; then there exists an $f$ such that $\|Q\|_{f,\rho,B}>0$. Hence, the seminorms separate $Q$'s.} topology generated by the seminorms
$$
Q\mapsto\|Q\|_{f,\rho,B}:=\left|\int_\Omega f(x)\d[Q(\rho,B)](x)\right|,\qquad
f\in C_0(\Omega),\;\rho\in\th,\;B\in\lk,
$$
then its topological subspace $\rm{Ball}(\hi,\,\ki,\,\Omega)$ is compact. 

Of course, next we {\it would} like to show that its subspace $\In\big(\bo\Omega,\,\ki,\,\hi\big)$ is closed and thus compact, but this does not hold in general: from Remark \ref{Remark1} we see that, in the case $\Omega=\R$, the sequence of instruments $\M_n(X,B)=\delta_n(X)\Phi_n(B)$ (where $\Phi_n:\,\lk\to\lh$ are channels) converges to the zero map. Another problem is that the set of those bounded bilinear forms $Q$ such that $B\mapsto [Q(\rho,B)](X)$ is normal\footnote{If it is normal then there exists an $\mathcal I(X,\rho)\in\tk$ such that $[Q(\rho,B)](X)=\tr{\mathcal I(X,\rho)B}$ for all $B\in\lk$.} for all $\rho\in\th$ and $X\in\Sigma$ is not typically closed. The CP-condition (CP) above is not a problem, since if $\sum_{s,t=1}^n Q_i(\kb{\psi_t}{\psi_s},B_s^*B_t)\ge 0$ for all $i\in\mathcal I$ then the limit $\sum_{s,t=1}^n Q(\kb{\psi_t}{\psi_s},B_s^*B_t)\ge 0$ also.\footnote{Since the converging net of nonnegative measures/numbers must converge to the nonnegative measure/number.} So the problem (NO) is that ${\rm prob}(\Omega)$ is not compact (unless $\Omega$ is compact). To overcome this difficulty, we replace $\Omega$ with its compactification $\ov\Omega$ everywhere in the above calculations and replace, in the cartesian product \eqref{cart}, the spaces ${\rm ball}_ {\|\rho\|_1\|\id_\ki\|}(\ov\Omega)$, $\rho\in\sh$, with the same compact space ${\rm prob}(\ov\Omega)$; in this way, we immediately obtain a compact cartesian product space whose convex subset $\In_{\rm H}\big(\bo{\ov\Omega},\,\ki,\,\hi\big)$ of Heisenberg channels is compact. Moreover, one can view $\In_{\rm H}\big(\bo{\Omega},\,\ki,\,\hi\big)$ as a convex subset of $\In_{\rm H}\big(\bo{\ov\Omega},\,\ki,\,\hi\big)$ by extending each $\M\in\In_{\rm H}\big(\bo{\Omega},\,\ki,\,\hi\big)$ via $\M(\{\infty\},B):=0$ for all $B\in\lk$. Hence, 
$$
\In_{\rm H}\big(\bo{\Omega},\,\ki,\,\hi\big)=\big\{\M\in\In_{\rm H}\big(\bo{\ov\Omega},\,\ki,\,\hi\big)\,\big|\,\mu^{\M}_{\rho,{\id_\ki}}(\{\infty\})=0\big\}
$$
where (a fixed) $\rho\in\sh$ has only positive eigenvalues; note that $\mu^{\M}_{\rho,{\id_\ki}}\in\mathrm{prob}(\Omega)\subseteq{\rm prob}(\ov\Omega)$. We next show that $\In_{\rm H}\big(\mc B(\Omega),\hi,\ki\big)$ is measurable; compare this result to Lemma 6 of \cite{ChDASc}.

\begin{lemma}\label{lemma:Gdelta}
The set $\In_{\rm H}\big(\mc B(\Omega),\hi,\ki\big)$ is a $G_\delta$-subset of the set of bounded bilinear forms $Q:\th\times\lk\to{\rm ca}(\Omega)$ and, as such, Borel measurable.
\end{lemma}

\begin{proof}
Define the sets
$$
G_n:=\left\{\M\in\In_{\rm H}\big(\mc B(\overline{\Omega}),\hi,\ki\big)\,\middle|\,\M(\{\infty\},\id_\ki)<\frac{1}{n}\id_\hi\right\}
$$
for $n\in\N$. Let $(\M_i)_{i\in\mc I}$ be a net in the complement $G_n^c$ of $G_n$ converging to $\M$ which we already know to be a Heisenberg channel on $\overline{\Omega}$. Clearly, the claim follows by proving that $\M\in G_n^c$ which we now go on to show. Using the fact that, for all $i\in\mc I$ and $X\in\mc B(\Omega)$, the spectrum of $\M_i(X,\id_\ki)$ contains elements at most $1-1/n$, we find a positive $\rho\in\th$ such that $\tr{\rho}=1$ so that
$$
\int_\Omega f\,\d\mu^{\M_i}_{\rho,\id_\ki}\leq 1-\frac{1}{n},\qquad
f\in C_0(\Omega),\quad 0\leq f\leq\chi_\Omega
$$
Thus,
\begin{align*}
\tr{\rho\M_i(\Omega,\id_\ki)}=&\sup\left\{\int_\Omega f\,\d\mu^{\M_i}_{\rho,\id_\ki}\,\middle|\,f\in C_0(\Omega),\; 0\leq f\leq\chi_\Omega\right\}\leq 1-\frac{1}{n}.
\end{align*}
Using the convergence, this implies $\int_\Omega f\,\d\mu^{\M}_{\rho,\id_\ki}\leq 1-1/n$ for all $f\in C_0(\Omega)$ such that $0\leq f\leq\chi_\Omega$. This means that $\tr{\rho\M(\Omega,\id_\ki)}\leq 1-1/n$, i.e.,\ $\tr{\rho\M(\{\infty\},\id_\ki)}\geq 1/n$. Thus, the spectrum of $\M(\{\infty\},\id_\ki)$ contains elements at least $1/n$, i.e.,\ $\M\in G_n^c$.
\end{proof}

In conclusion,
\begin{itemize}
\item the topology of the locally convex Hausdorff space ${\rm Bil}\big(\th,\lk;\mathrm{ca}(\ov\Omega)\big)$ of bounded bilinear maps $Q:\,\th\times\lk\to\mathrm{ca}(\ov\Omega)$ is generated by the seminorms
$$
Q\mapsto\|Q\|_{g,\rho,B}:=\left|\int_\Omega g(x)\d[Q(\rho,B)](x)\right|,\qquad
g\in C(\ov\Omega),\;\rho\in\th,\;B\in\lk.
$$
\item $\In_{\rm H}\big(\bo{\ov\Omega},\,\ki,\,\hi\big)$ is its compact convex subset and we denote the Borel $\sigma$-algebra of its (subspace) topology by $\bo{\ov\In}$.
\item $\In_{\rm H}\big(\bo{\Omega},\,\ki,\,\hi\big)$ is its measurable convex subset.\footnote{Note that the sets $\ext\In_{\rm H}\big(\bo{\Omega},\,\ki,\,\hi\big)$ and $\ext\In_{\rm H}\big(\bo{\ov\Omega},\,\ki,\,\hi\big)$ are not necessarily measurable.}
\item A converging net $i\mapsto \M_i$ of $\In_{\rm H}\big(\bo{\Omega},\,\ki,\,\hi\big)$ may converge to (an unphysical instrument) $\M\in\In_{\rm H}\big(\bo{\ov\Omega},\,\ki,\,\hi\big)$; the convergence means that
\begin{eqnarray*}
\lim_{i\in\mathcal I}\int_\Omega g(x)\tr{\rho\M_i(\d x,B)}&=&\int_{\Omega\cup\{\infty\}} g(x)\tr{\rho\M(\d x,B)} \\
&=&\int_{\Omega} g(x)\tr{\rho\M(\d x,B)}+ g(\infty)\tr{\rho\M(\{\infty\},B)}
\end{eqnarray*}
for all $g\in C(\ov\Omega)$, $\rho\in\th$, and $B\in\lk$.\footnote{For example, take a sequence of channels $\Phi_n$ which converges to a CP channel $\Phi$ (i.e.\ $\lim_{n\to\infty}\tr{\rho \Phi_n(B)}\equiv \tr{\rho \Phi(B)}$). Then the instruments $\M_n(X,B)=\delta_n(X)\Phi_n(B)$ converge to an unphysical instrument $\M(X,B)=\delta_\infty(X)\Phi(B)$, $X\in\bo{\R\cup\{\infty\}}$, $B\in\lk$.}
\item For all $k=1,\ldots,N$, let $\M_k\in\In_{\rm H}\big(\bo{\ov\Omega},\,\ki,\,\hi\big)$ and $\la_k> 0$ be such that $\sum_{k=1}^N\la_k=1$. Denote $\M=\sum_{k=1}^N\la_k\M_k\in\In_{\rm H}\big(\bo{\ov\Omega},\,\ki,\,\hi\big)$. Then
\begin{itemize}
\item $\{\M_k\}_{k=1}^N\subset\In_{\rm H}\big(\bo{\Omega},\,\ki,\,\hi\big)$ yields $\M\in\In_{\rm H}\big(\bo{\Omega},\,\ki,\,\hi\big)$,
\item if some $\M_k\notin\In_{\rm H}\big(\bo{\Omega},\,\ki,\,\hi\big)$ then $\M\notin\In_{\rm H}\big(\bo{\Omega},\,\ki,\,\hi\big)$.
\end{itemize}
\item $\ext\In_{\rm H}\big(\bo{\Omega},\,\ki,\,\hi\big)\subseteq\ext\In_{\rm H}\big(\bo{\ov\Omega},\,\ki,\,\hi\big)$.
\item $\ext\In_{\rm H}\big(\bo{\Omega},\,\ki,\,\hi\big)=\In_{\rm H}\big(\bo{\Omega},\,\ki,\,\hi\big)\cap\ext\In_{\rm H}\big(\bo{\ov\Omega},\,\ki,\,\hi\big)$.
\end{itemize}
The Choquet--Bishop--de Leeuw theorem says that {\it any $\M\in\In_{\rm H}\big(\bo{\ov\Omega},\,\ki,\,\hi\big)=:\ov\In$ can be represented by a probability boundary measure} \cite[Theorem I.4.8, p.\ 36]{Alfsen}.
This means that there is a (possibly nonunique) probability measure $p_\M:\,\bo{\ov\In}\to[0,1]$ such that its support is a subset of the {\it closure} of $\ext\In_{\rm H}\big(\bo{\ov\Omega},\,\ki,\,\hi\big)$, and
$
\M=\int_{\ov\In} \hE\d p_\M(\hE)
$ 
that is, $\M$ is the {\it barycenter} of $p_\M$; the above (weak) integral means that
$F(\M)=\int_{\ov\In} F(\hE)\d p_\M(\hE)$ for all continuous (real)\footnote{Any complex (topological) vector space is trivially a real (topological) vector space, any complex linear function is real linear, and the real and imaginary parts of a complex valued continuous function are real linear and continuous.} linear (real valued) functions $F$ on ${\rm Bil}\big(\th,\lk;\mathrm{ca}(\ov\Omega)\big)$.
Especially,
$$
\int_{\ov\Omega} g(x)\tr{\rho\M(\d x,B)}=
\int_{\ov\In}\int_{\ov\Omega} g(x)\tr{\rho\hE(\d x,B)}\d p_\M(\hE)
$$
for all $g\in C(\ov\Omega)$, $\rho\in\th$, and $B\in\lk$.

The closure of the set of extreme points of the set of instruments can be `huge', so the above result might not be very useful: if  $\M$ belongs to the closure  then one can trivially choose $p_\M$ to be the Dirac measure $\delta_\M$ at $\M$. In order to apply the more usable Choquet theorem which associates to a (quantum) instrument a probability measure which is supported already by the set of extreme points, not its closure, we have to show that $\In\big(\mc B(\Omega),\ki,\hi\big)$ is metrizable with respect to the subspace topology defined above. We do this by demonstrating that the subspace topology is defined by a countable number of seminorms.

Since $\hi$ is separable, $\th$ is separable (w.r.t.\ the trace-norm). If $(A_i)_{i\in\mathcal I}\subseteq\lh$ is norm-bounded (i.e.\ $\sup_{i\in\mathcal I}\|A_i\|<\infty$) and $A\in\lh$ then $\lim_{i\in\mathcal I}\<h_m|A_i h_n\>=\<h_m|A h_n\>$ for all $m,\,n$ if and only if $\lim_{i\in\mathcal I}\tr{\rho A_i}=\tr{\rho A}$ for all $\rho\in\th$; here $\{h_n\}_{n=1}^{\dim\hi}$ is an orthonormal basis of $\hi$. Similarly, since $C(\ov\Omega)$ is separable having a dense sequence $\{g_k\}_{k=1}^\infty$, one gets, for $(\mu_i)_{i\in\mathcal I}\subseteq\mathrm{ball}_r(\ov\Omega)$ and $\mu\in\mathrm{ball}_r(\ov\Omega)$, that $\lim_{i\in\mathcal I}\int_{\ov\Omega}g_k(x)\d\mu_i(x)=\int_{\ov\Omega}g_k(x)\d\mu(x)$ for all $k$ if and only if $\lim_{i\in\mathcal I}\int_{\ov\Omega}g(x)\d\mu_i(x)=\int_{\ov\Omega}g(x)\d\mu(x)$ for all $g\in C(\ov\Omega)$. These facts show that $\lim_{i\in\mathcal I}\M_i=\M$ if
$$
\lim_{i\in\mathcal I}\<h_m|L(g_k,\Mo_i^B)h_n\>=\<h_m|L(g_k,\Mo_i^B)h_n\>
$$
for all $m,\,n,\,k$ and $B\in\lk$.\footnote{Note that $\sup_{i\in\mathcal I}\|L(g_k,\Mo_i^B)\|\le\|g_k\|_\infty\|B\|<\infty$.}

Suppose then that the `output space' $\ki$ is {\it finite dimensional} and $\{k_s\}_{s=1}^{\dim\ki}$ is its orthonormal basis.\footnote{Note that if $\dim\ki=\infty$ then $\lk$ is not separable w.r.t.\ the operator norm.} Now all  Heisenberg instruments are actually  quantum instruments since all the relevant topologies of the now finite-dimensional vector space $\lk$, including the ultraweak topology, coincide with the Euclidean topology. This means that
\begin{eqnarray*}
\ov\In &:= \In_{\rm H}\big(\mc B(\overline{\Omega}),\ki,\hi\big)&=\In\big(\mc B(\overline{\Omega}),\ki,\hi\big),\\
\In &:=\In_{\rm H}\big(\mc B(\Omega),\ki,\hi\big)&=\In\big(\mc B(\Omega),\ki,\hi\big).
\end{eqnarray*}
We denote the sets of extreme points of the above sets by $\ext\ov\In$ and $\ext\In$ respectively. By writing $B=\sum_{s,t=1}^{\dim\ki}B_{st}\kb{k_s}{k_t}$ (where $B_{st}:=\<k_s|Bk_t\>$) one gets
$$
\lim_{i\in\mathcal I}\<h_m|L(g_k,\Mo_i^B)h_n\>
=\sum_{s,t=1}^{\dim\ki}B_{st}\;
\lim_{i\in\mathcal I}\<h_m|L(g_k,\Mo_i^{\kb{k_s}{k_t}})h_n\>
$$
so already the countable number of seminorms $\|\;\cdot\;\|_{g_k,\kb{h_n}{h_m},\kb{k_s}{k_t}}$ determine the topology of $\In\big(\bo{\ov\Omega},\,\ki,\,\hi\big)$, that is, it is pseudometrizable. Since it is also Hausdorff, $\In\big(\bo{\ov\Omega},\,\ki,\,\hi\big)$ is a metrizable (compact convex) set and we may apply the Choquet theorem.

\begin{theorem}\label{thm:bary}
Suppose that $\Omega$ is a locally compact second countable Hausdorff space, $\ki$ a finite dimensional Hilbert space, $\hi$ a separable Hilbert space, and $\M\in\ov\In$. Then $\ext\ov\In$ is measurable and there exists a probability measure $p_\M:\,\mc B(\ov\In)\to[0,1]$ such that $p_\M(\ext\ov\In)=1$ and
\begin{equation}\label{xyz}
\int_{\ov\Omega} g(x)\tr{\rho\M(\d x,B)}=
\int_{\ov\In}\int_{\ov\Omega} g(x)\tr{\rho\hE(\d x,B)}\d p_\M(\hE)
\end{equation}
for all $g\in C(\ov\Omega)$, $\rho\in\th$, and $B\in\lk$. If $\M\in\In$, i.e.,\ $\M$ is a physical instrument, then $p_\M(\ext\In)=1$, so that
$$
\int_\Omega f(x)\tr{\rho\M(\d x,B)}=\int_{\ext\In}\int_\Omega f(x)\tr{\rho\mc E(\d x,B)}\d p_\M(\mc E)
$$
for all $f\in C_0(\Omega)$, $\rho\in\th$, and $B\in\lk$.
\end{theorem}

\begin{proof}
The first half of the claim follows directly from the fact demonstrated above that $\In\big(\bo{\ov\Omega},\,\ki,\,\hi\big)$ is a metrizable, convex, and compact and the Choquet theorem \cite[Corollary I.4.9, p.\ 36]{Alfsen}. The last claim follows immediately from $\ext\In=\In\cap\ext\ov\In$ by proving that $p_\M(\In)=1$. Let us show this. Suppose that $\M\in\In\big(\bo{\Omega},\,\ki,\,\hi\big)$ and put $B=\id_\ki$ and a state $\rho$ with positive eigenvalues in \eqref{xyz} to get
$$
\int_{\ov\Omega} g(x)\d\mu^{\M}_{\rho,\id_\ki}(x)=
\int_{\ov\In}\int_{\ov\Omega} g(x)\d\mu^{\hE}_{\rho,\id_\ki}(x)\d p_\M(\hE).
$$
Since the set $\ov\In_n:=\big\{\hE\in\ov\In\,\big|\,\mu^{\hE}_{\rho,\id_\ki}(\Omega)\le1-n^{-1}\big\}$ is closed (see the proof of Lemma \ref{lemma:Gdelta}) one can write, for any $f\in C_0(\Omega)$, $0\le f\le\CHI\Omega$,
\begin{eqnarray*}
\int_{\ov\Omega} f(x)\d\mu^{\M}_{\rho,\id_\ki}(x)&=&
\int_{\ov\In_n}\int_{\ov\Omega} f(x)\d\mu^{\hE}_{\rho,\id_\ki}(x)\d p_\M(\hE)+
\int_{\ov\In\setminus\ov\In_n}\int_{\ov\Omega} f(x)\d\mu^{\hE}_{\rho,\id_\ki}(x)\d p_\M(\hE) \\
&\le&\left(1-\frac1n\right)p_\M(\ov\In_n)+p_\M(\ov\In\setminus\ov\In_n)=1-\frac1n p_\M(\ov\In_n).
\end{eqnarray*}
Hence,
$$
1=\mu^{\M}_{\rho,\id_\ki}(\Omega)=\sup\left\{\int_\Omega f(x)\d\mu^{\M}_{\rho,\id_\ki}(x)\,\Big|\,f\in C_0(\Omega),\;0\le f\le\CHI\Omega\right\}\le1-\frac1n p_\M(\In_n),
$$
that is, $p_\M(\ov\In_n)=0$ for all $n=1,2,\ldots,$ showing that $p_\M(\In)=1$. 
\end{proof}

\begin{example}\rm
In the context of the Theorem \ref{thm:bary} we get, for example:
\begin{itemize}
\item(Output space $\ki=\C$.) For any POVM $\Mo:\,\bo\Omega\to\lh$ we have 
$$
\int_\Omega f(x)\d\Mo(x)=\int_{\rm POVMs} \int_\Omega f(x)\d\sfe(x)\d p_\Mo(\sfe), \qquad f\in C_0(\Omega), 
$$
where $p_\Mo$ is supported by extreme POVMs. Since we may have $\dim\hi=\infty$ this result is a generalization of \cite[Theorem 5]{Ali} and \cite[Corollary 4]{ChDASc2}.
\item(Value space $\Omega=\{1\}$.) For any quantum channel $\Phi:\,\lk\to\lh$ we have $\Phi(B)=\int_{\rm Channels}\Theta(B)\d p_\Phi(\Theta)$, $B\in\lk$, where $p_\Phi$ is supported by extreme channels. In the Schr\"odinger picture, $\Phi_*(\rho)=\int_{\rm Channels}\Theta_*(\rho)\d p_\Phi(\Theta)$ for all $\rho\in\th$. Recall that we require that $\dim\ki<\infty$. For example, if $\dim\hi\le \dim\ki$ then any isometry $J:\,\hi\to\ki$, $J^*J=\id_\hi$, defines an isometry channel $\Theta_J(B):=J^*BJ$ which is extreme \cite{instru1}. Hence, a random isometry channel $\Phi(B)=\int_{\rm Isometries}\Theta_J(B)\d p_\Phi(\Theta_J)$ can be viewed as a barycenter. Especially in the case $\dim\hi=\dim\ki$ we get the random unitary channels.
\item (Value space $\Omega=\{1\}$ and $\hi=\C$.) Now any $\Phi_*$ is just a state $\sigma\in\mathcal S(\ki)$ so we get $\sigma=\int_{\rm Pure}\kb{\psi}{\psi}\d p_\sigma(\kb{\psi}{\psi})$. This shows that $p_\sigma$ is not unique since it can always be replaced by a discrete probability measure which gives the eigendecomposition $\sigma=\sum_i \kb{\psi_i}{\psi_i}p_i$, where $\sum_i p_i=1$, $p_i\ge 0$. 
\item (Input and output spaces are $\C$.) POVMs are now probability measures $p$ and extreme probability measures are Dirac measures $\delta_x$, i.e.\ points $x$. One can write $\int_\Omega f(x)\d p(x)=\int_{\Omega}\int_\Omega f(y)\d\delta_x(y)\d\tilde p(\delta_x)=\int_\Omega f(x)\d\tilde p(\delta_x)$ so one can identify $\d\tilde p(\delta_x)$ with $\d p(x)$.
\end{itemize}
\end{example}

\subsection{Example: Qubit effects}

Assume as above $\ki=\C$ and also $\Omega=\{0,1\}$ so we are dealing with effects $E$ of $\hi$ (now the binary POVM is $\Mo(\{1\})=E$, $\Mo(\{0\})=\id_\hi-E$). We have $E=\int_{\rm Projections} P\d p_E(P)$. Let us check the qubit case $\hi=\C^2$: Any effect $E$ can be written in the form
\begin{equation}\label{E}
E=\frac12\sum_{\mu=0}^3 e^\mu\sigma_\mu= \frac12(e^0\sigma_0+\ve\cdot\vsigma)
=\frac12\begin{pmatrix}
e^0+e^3 & e^1-i e^2 \\
e^1+i e^2 & e^0-e^3
\end{pmatrix}
\end{equation}
where $(e^0,e^1,e^2,e^3)=(e^0,\ve)\in \R^4$, $\|\ve\|:=\sqrt{(e^1)^2+(e^2)^2+(e^3)^2}\le\min\{e^0,2-e^0\}$
\cite{Measurement}. In particular, $e^\mu=\tr{E\sigma_\mu}$, $\mu=0,1,2,3$, $e^0\in[0,2]$, and $|e^j|\le\|\ve\|\le1$, $j=1,2,3$. The eigenvalues of $E$ are $$\lambda^E_\pm:=\frac{1}{2}(e^0\pm\|\ve\|)\in[0,1]$$ so that $E$ is of rank 1 if and only if $e^0=\|\ve\|\ne0$. Especially, $E$ is a rank-1 (resp.\ rank-2) projection exactly when $e^0=\|\ve\|=1$ (resp.\ $e^0=2$ and $\|\ve\|=0$, i.e.\ $E=\sigma_0$). Denote by $\mc P$ the set of projections on $\C^2$ and by $\mc P_1$ the subset of rank-1 projections. Hence, define $\vp(\theta,\fii):=(\cos\fii\sin\theta,\,\sin\fii\sin\theta,\,\cos\theta)$ to get
\begin{eqnarray*}
E&=&\int_{\mc P} P\d p_E(P)=\int_{\mc P_1} P\d p_E(P)+\sigma_0 p_E(\{\sigma_0\}) \\
&=&\int_0^{\pi}\int_{0}^{2\pi} \frac12(\sigma_0+\vp(\theta,\fii)\cdot\vsigma)\d p_E(\theta,\fii)+\sigma_0 p_E(\{\sigma_0\}) \\
&=&
\frac12\underbrace{\left[1+p_E(\{\sigma_0\})-p_E(\{0\})\right]}_{=\;e^0}\sigma_0+
\frac12\underbrace{\left[
\int_0^{\pi}\int_{0}^{2\pi} \vp(\theta,\fii)\d p_E(\theta,\fii)
\right]}_{=\;\ve}\cdot\vsigma.
\end{eqnarray*}
Note that the positive measure on the Bloch sphere is not necessarily normalized:
$$
\int_0^{\pi}\int_{0}^{2\pi} \d p_E(\theta,\fii)=1-p_E(\{\sigma_0\})-p_E(\{0\}).
$$
For a given $E$ as in \eqref{E} with $\ve\ne0$, define the projection $P_E:=\frac12\big(\sigma_0+\|\ve\|^{-1}\ve\cdot\vsigma\big)$ and choose
$$
p_E
=\lambda^E_-\delta_{\sigma_0}+(1-\lambda^E_+)\delta_0+(\lambda^E_+-\lambda^E_-)\delta_{P_E}
=\frac{1}{2}(e^0-\|\ve\|)\delta_{\sigma_0}+\left[1-\frac{1}{2}(e^0+\|\ve\|)\right]\delta_0+\|\ve\|\delta_{P_E}.
$$
One may easily check that $E$ is represented by $p_E$, i.e.\ $E=\int P\d p_E(P)$.

\subsection{Example: Qubit channels}

Recall that a quantum channel $\Phi:\lk\to\lh$ (with finite-dimensional $\hi$ and $\ki$) is extreme in the convex set of channels $\lh\to\lk$ if it has a minimal set $\{K_\ell\}_{\ell\in L}$ of Kraus operators such that the set $\{K_k^*K_\ell\}_{k,\ell\in L}$ is linearly independent. The requirement that $\{K_\ell\}_{\ell\in L}$ be a minimal set of Kraus operators for $\Phi$ means that
$$
\Phi(B)=\sum_{\ell\in L}K_\ell^*BK_\ell
$$
for all $B\in\lk$ and that $\{K_\ell\}_{\ell\in L}$ is linearly independent. If the space $\hi$ is infinite dimensional, the extremality condition is essentially the same as above but with slight modifications \cite{Tsuikkis}.

We call quantum channels $\Phi:\lk\to\mc L(\C^2)$ (in Heisenberg picture) {\it qubit channels} with output space $\ki$. Next we will characterize the extreme points of this convex set with a fixed output space $\ki$. Suppose that $\{K_\ell\}_{\ell\in L}$ is a minimal set of Kraus operators for an extreme qubit channel $\Phi$ with output space $\mc K$. Since $\mc L(\C^2)$ is 4-dimensional, we must have $|L|\in\{1,2\}$, so that the set $\{K_k^*K_\ell\}_{k,\ell\in L}$ can be linearly independent. Thus, this set is either a singleton (consisting of a single isometry) or $\{K_0,K_1\}$. The first case is rather trivial, so let us concentrate on the second case. We now have
$$
K_0^*K_0+K_1^*K_1=\id.
$$
This means that $K_0^*K_0$ cannot be a multiple of identity because then also $K_1^*K_1$ would be a multiple of identity too and the set $\{K_0^*K_0,K_1^*K_1\}$ would already be linearly dependent. Thus, $K_0^*K_0$ must have a non-degenerate spectrum and, thus, a unique eigenbasis (up to phase factors). It turns out that the remaining extremality condition is quite mild, as the following result tells us.

\begin{proposition}\label{prop:extqubit}
Let $\Phi:\lk\to\mc L(\C^2)$ be a qubit channel. For this channel to be extreme within the set of qubit channels with output space $\ki$, it is necessary that $\Phi$ has a minimal set of Kraus operators consisting of
\begin{itemize}
\item[(i)] only one isometry $K_0$ or
\item[(ii)] two non-zero operators $K_0$ and $K_1$.
\end{itemize}
In case (i), this is also a sufficient condition for the extremality of $\Phi$. For $\Phi$ of case (ii) to be extreme, it is necessary that $K_0^*K_0$ has a non-degenerate spectrum and, thus, a unique eigenbasis $\{|0\>,|1\>\}$ (up to phase factors). This qubit channel is now extreme if and only if
\begin{equation}\label{eq:ehtoE}
|\<0|K_0^*K_1|1\>|\neq |\<1|K_0^*K_1|0\>|.
\end{equation}
\end{proposition}

\begin{proof}
As we already observed, for $\Phi$ to be extreme within the set of qubit channels with output space $\ki$, it must have a minimal set of Kraus operators containing only up to two operators. In the case of a single operator $K_0$, $K_0$ must be an isometry so that the normalization condition $K_0^*K_0=\id$ holds. The singleton $\{K_0^*K_0\}=\{\id\}$ is automatically linearly independent, so in this case $\Phi$ is extreme. Let us concentrate to the case (ii) for the remainder of this proof and let $\{K_0,K_1\}$ be a set of Kraus operators for $\Phi$. We have already observed that $K_0^*K_0$ (or, equivalently, $K_1^*K_1$) must have a non-degenerate spectrum and, thus, a unique eigenbasis $\{|0\>,|1\>\}$ (up to phase factors) so that $\Phi$ is extreme. Let us thus assume that
\begin{align*}
E&:=K_0^*K_0=a|0\>\<0|+b|1\>\<1|,\\
\id-E&:=K_1^*K_1=(1-a)|0\>\<0|+(1-b)|1\>\<1|,
\end{align*}
where $a,b\in[0,1]$, $a\neq b$. Let us first assume that both $E$ and $\id-E$ are of rank 1. Naturally, this means that one of the eigenvalues, say $a$, is 1, so that $b=0$. Using the polar decomposition of operators, this means that there are unit vectors $w,w'\in\ki$ such that $K_0=|w\>\<0|$ and $K_1=|w'\>\<1|$. We now have
$$
\{K_i^*K_j\}_{i,j=0}^1=\{|0\>\<0|,\<w|w'\>|0\>\<1|,\<w'|w\>|1\>\<0|,|1\>\<1|\}.
$$
This set is linearly independent if and only if $\<w|w'\>\neq0$ which is evidently equivalent with \eqref{eq:ehtoE}.

Let us now assume that at least one of the operators $E$ or $\id-E$ is of rank 2 (full rank). We may freely assume that $E$ is of rank 2 (i.e.,\ $a,b>0$) and, hence, has an inverse $E^{-1}$. Due to the polar decomposition, there are isometries $U_0$ and $U_1$ so that $K_0=U_0E^{1/2}$ and $K_1=U_1(\id-E)^{1/2}$. Denote $V:=U_0^*U_1$. We now have
$$
\{K_i^*K_j\}_{i,j=0}^1=\{E,E^{1/2}V(\id-E)^{1/2},(\id-E)^{1/2}V^*E^{1/2},\id-E\}.
$$
This set is linearly independent if and only if the set
$$
\mc S:=\{\id,VA^{1/2},A^{1/2}V^*,A\}
$$
is linearly independent where $A:=E^{-1}-\id$ is a positive operator with the eigenvalues $\lambda:=1/a-1$ and $\mu:=1/b-1$. This is seen by multiplying the set $\{K_i^*K_j\}_{i,j=0}^1$ on both sides with $E^{-1/2}$. For $\mc S$ to be linearly independent, the subset $\{\id,VA^{1/2},A\}$ must be linearly independent. Denote $v:=V|0\>$, $\tilde{v}:=V|1\>$. For $\alpha_i\in\C$, $i=1,2,3$, we have
\begin{align*}
\alpha_1 \id+\alpha_2 VA^{1/2}+\alpha_3 A&=(\alpha_1+\lambda\alpha_3)|0\>\<0|+(\alpha_1+\mu\alpha_3)|1\>\<1|+\sqrt{\lambda}\alpha_2|v\>\<0|+\sqrt{\mu}\alpha_2|\tilde{v}\>\<1|.
\end{align*}
Reading from the above and using the fact that $\lambda\neq\mu$, we have that $\{\id,VA^{1/2},A\}$ is linearly independent if and only if $\{|0\>\<0|,|v\>\<0|,|\tilde{v}\>\<1|,|1\>\<1|\}$ is linearly independent. For the latter, $\{|0\>\<0|,|v\>\<0|,|1\>\<1|\}$ must be linearly independent, implying that we cannot have $\<1|v\>=0$, i.e.,\ $\<1|V|0\>=0$. Indeed, if $\<1|v\>=0$, $v$ would have no contribution from $|1\>$ and would be a multiple of $|0\>$, whence the set $\{|0\>\<0|,|v\>\<0|,|1\>\<1|\}$ would be linearly dependent. Thus, $\<1|V|0\>\neq0$. Similarly, we see that $\<0|\tilde{v}\>\neq0$, i.e.,\ $\<0|V|1\>\neq0$ if $\{|0\>\<0|,|\tilde{v}\>\<1|,|1\>\<1|\}$ is to be linearly independent. Thus, for $\mc S$ to be linearly independent, we must have
\begin{equation}\label{eq:ehtoS}
\<1|V|0\>\neq0\neq\<0|V|1\>.
\end{equation}

Let now $\alpha_i\in\C$, $i=1,2,3,4$, be such that $\alpha_1\id+\alpha_2 VA^{1/2}+\alpha_3 A^{1/2}V^*+\alpha_4 A=0$. This is equivalent with
\begin{align}
0&=\big[\alpha_1+\sqrt{\lambda}(\<0|V|0\>\alpha_2+\<0|V^*|0\>\alpha_3)+\lambda\alpha_4\big]|0\>\<0|+\big[\sqrt{\mu}\<0|V|1\>\alpha_2+\sqrt{\lambda}\<0|V^*|1\>\alpha_3\big]|0\>\<1|\nonumber\\
&+\big[\sqrt{\lambda}\<1|V|0\>\alpha_2+\sqrt{\mu}\<1|V^*|0\>\alpha_3\big]|1\>\<0|+\big[\alpha_1+\sqrt{\mu}(\<1|V|1\>\alpha_2+\<1|V^*|1\>\alpha_3)+\mu\alpha_4\big]|1\>\<1|.\label{eq:linrton}
\end{align}
Checking the determinant condition, we see that the simultaneous vanishing of the factors of $|0\>\<1|$ and $|1\>\<0|$ in \eqref{eq:linrton} only allow for the trivial solution $\alpha_2=0=\alpha_3$ if and only if $\mu|\<0|V|1\>|^2\neq\lambda|\<1|V|0\>|^2$. This condition is immediately seen to coincide with \eqref{eq:ehtoE}. Thus, if \eqref{eq:ehtoE} holds, we get $\alpha_2=0=\alpha_3$ from the coefficients of $|0\>\<1|$ and $|1\>\<0|$ in \eqref{eq:linrton}, and what remains is $\alpha_1+\lambda\alpha_4=0$ and $\alpha_1+\mu\alpha_4=0$. Since $\lambda\neq\mu$, this can only happen if $\alpha_1=0=\alpha_4$. Thus, \eqref{eq:ehtoE} is a sufficient condition for the extremality of $\Phi$.

Let us now show the necessity of \eqref{eq:ehtoE} for the extremality of $\Phi$. Assume that \eqref{eq:ehtoE} does not hold, i.e.,\ $\mu|\<0|V|1\>|^2=\lambda|\<1|V|0\>|^2$. Let $\alpha_i\in\C$, $i=1,2,3,4$, be such that \eqref{eq:linrton} holds. Let us first deal with the case $\mu|\<0|V|1\>|^2=\lambda|\<1|V|0\>|^2=0$. In the case $\mu=0$, we must thus have $\lambda=0$ or $\<1|V|0\>=0$. In the first case, from \eqref{eq:linrton}, we get $\alpha_1=0$ but $\alpha_2$, $\alpha_3$, and $\alpha_4$ may be whatever, i.e.,\ $\mc S$ is not linearly independent and $\Phi$ is not extreme. In the second case, we have already seen when deriving \eqref{eq:ehtoS}, that $\Phi$ cannot be extreme. The case $\lambda=0$ is treated similarly. The cases $\<0|V|1\>=0$ or $\<1|V|0\>=0$ directly lead to $\Phi$ not being extreme, as we have seen when deriving \eqref{eq:ehtoS}. Thus, we may assume that $\mu|\<0|V|1\>|^2=\lambda|\<1|V|0\>|^2$ where all the terms $\lambda$, $\mu$, $\<0|V|1\>$, and $\<1|V|0\>$ are non-zero. Moreover, we may assume $\lambda\neq\mu$ because $\mu=\lambda$ leads to $\{K_0^*K_0,K_1^*K_1\}$ already being linearly dependent. Using again the equations arising from the coefficients of $|0\>\<1|$ and $|1\>\<0|$ in \eqref{eq:linrton}, we have
$$
\alpha_3=-\sqrt{\frac{\mu}{\lambda}}\frac{\<0|V|1\>}{\<0|V^*|1\>}\alpha_2=:c\alpha_2.
$$
Denoting $L:=\<0|V|0\>+\<0|V^*|0\>c$ and $M:=\<1|V|1\>+\<1|V^*|1\>c$, we obtain from the coefficients of $|0\>\<0|$ and $|1\>\<1|$ in \eqref{eq:linrton}
\begin{equation}\label{eq:yhtaloryhma}
\left\{\begin{array}{rcl}
\alpha_1+\sqrt{\lambda}L\alpha_2+\lambda\alpha_4&=&0,\\
\alpha_1+\sqrt{\mu}M\alpha_2+\mu\alpha_4&=&0.
\end{array}\right.
\end{equation}
Subtracting these equalities and solving for $\alpha_4$, we get
$$
\alpha_4=\frac{\sqrt{\mu}M-\sqrt{\lambda}L}{\lambda-\mu}\alpha_2.
$$
Recall that $\lambda\neq\mu$, so this makes sense. Substituting this in \eqref{eq:yhtaloryhma}, gives two equivalent equations, as one easily checks, and they give
$$
\alpha_1=-\left(\sqrt{\mu}M+\mu\frac{\sqrt{\mu}M-\sqrt{\lambda}L}{\lambda-\mu}\right)\alpha_2=:M\alpha_2.
$$
All in all, by letting $\alpha_2\neq0$, we have $\alpha_3=c\alpha_2\neq0$ and, as we have seen, also $\alpha_1$ and $\alpha_4$ can be chosen to be multiples of $\alpha_2$, so that $\alpha_1\id+\alpha_2 VA^{1/2}+\alpha_2 A^{1/2}V^*+\alpha_4 A=0$, implying that $\mc S$ is linearly dependent, i.e.,\ $\Phi$ is not extreme. The factors $\alpha_1$ and $\alpha_4$ may vanish, but at least we can choose $\alpha_2\neq0\neq\alpha_3$ which is enough.
\end{proof}

Let us write the condition \eqref{eq:ehtoE} in a slightly different form. Let us assume that we are in the situation (ii) of Proposition \ref{prop:extqubit}. Using the polar decomposition, we find $a,b\in[0,1]$, a basis $\{|0\>,|1\>\}$ of $\C^2$, and unit vectors $w,w^\perp,\tilde{w},\tilde{w}^\perp\in\ki$ such that $\<w|w^\perp\>=0=\<\tilde{w}|\tilde{w}^\perp\>$ and
$$
K_0=\sqrt{a}|w\>\<0|+\sqrt{b}|w^\perp\>\<1|,\quad K_1=\sqrt{1-a}|\tilde{w}\>\<0|+\sqrt{1-b}|\tilde{w}^\perp\>\<1|.
$$
We now have
\begin{align*}
K_0^*K_1&=\sqrt{a(1-a)}\<w|\tilde{w}\>|0\>\<0|+\sqrt{a(1-b)}\<w|\tilde{w}^\perp\>|0\>\<1|\\
&+\sqrt{b(1-a)}\<w^\perp|\tilde{w}\>|1\>\<0|+\sqrt{b(1-b)}\<w^\perp|\tilde{w}^\perp\>|1\>\<1|.
\end{align*}
From this we find that \eqref{eq:ehtoE} is equivalent with
$$
\sqrt{a(1-b)}|\<w|\tilde{w}^\perp\>|\neq\sqrt{b(1-a)}|\<w^\perp|\tilde{w}\>|.
$$ 
Let us now assume $\mc K=\C^2$ because this is the simplest case. Let us fix a unit-vector-valued map $\mb S^2\ni\bm{p}\mapsto v_{\bm{p}}\in\C^2$ such that $|v_{\bm{p}}\>\<v_{\bm{p}}|=(1/2)(\id+\bm{p}\cdot\bm{\sigma})$, i.e.,\ we fix a section from the projective space of $\C^2$ into $\C^2$. Naturally, $\{v_{\bm{p}},v_{-\bm{p}}\}$ is an orthonormal basis of $\C^2$ and any orthonormal basis of $\C^2$ is of the form $\{e^{i\tj_1}v_{\bm{p}},e^{i\tj_2}v_{-\bm{p}}\}$ for some $\bm{p}\in\mb S^2$ and $\tj_1,\tj_2\in[0,2\pi)$. Define, for any $\bm{p},\bm{q}\in\mb S^2$, $a,b\in[0,1]$, and $\tj_1,\tj_2\in[0,2\pi)$,
$$
K_{\bm{p},\bm{q},a,b,\tj_1,\tj_2}:=e^{i\tj_1}\sqrt{a}|v_{\bm{q}}\>\<v_{\bm{p}}|+e^{i\tj_2}\sqrt{b}|v_{-\bm{q}}\>\<v_{-\bm{p}}|.
$$
We may now associate to any channel $\Phi:\mc L(\C^2)\to\mc L(\C^2)$ with up to two Kraus operators a set of Kraus operators of the form $\{K_{\bm{p},\bm{q},a,b,\tj_1,\tj_2},K_{\bm{p},\bm{r},1-a,1-b,\fii_1,\fii_2}\}$ where $\bm{p},\bm{q},\bm{r}\in\mb S^2$, $a,b\in[0,1]$, and $\tj_1,\tj_2,\fii_1,\fii_2\in[0,2\pi)$, and we denote this channel by $\Phi_{\bm{p},\bm{q},\bm{r},a,b,\tj_1,\tj_2,\fii_1,\fii_2}$. For this channel, the condition \eqref{eq:ehtoE} reads
$$
a(1-b)(1-\bm{q}\cdot\bm{r})\neq b(1-a)(1-\bm{q}\cdot\bm{r})
$$
which is satisfied if and only if $a\neq b$ and $\bm{q}\neq\bm{r}$. This means that the set of extreme points of the set of channels $\Phi:\mc L(\C^2)\to\mc L(\C^2)$ coincides with the union of two sets: the set of channels with only one Kraus operator, i.e.,\ unitary channels $B\mapsto U^*BU$, and the set of those $\Phi_{\bm{p},\bm{q},\bm{r},a,b,\tj_1,\tj_2,\fii_1,\fii_2}$ such that $a\neq b$ and $\bm{q}\neq\bm{r}$.

Let us denote the set of those $(\bm{p},\bm{q},\bm{r},a,b,\tj_1,\tj_2,\fii_1,\fii_2)\in(\mb S^2)^3\times[0,1]^2\times[0,2\pi)^4$ such that $\bm{q}\neq\bm{r}$ and $a\neq b$ by $\mc C$ and the group of unitary ($2\times 2$)-matrices by $U(2)$. According to Theorem \ref{thm:bary}, any channel $\Phi:\mc L(\C^2)\to\mc L(\C^2)$ can now be expressed with some $t\in[0,1]$ and probability measures $p_1:\mc B\big(U(2)\big)\to[0,1]$ and $p_2:\mc B(\mc C)\to[0,1]$ through
$$
\Phi(B)=t\int_{U(2)}U^*BU\,\d p_1(U)+(1-t)\int_{\mc C} \Phi_{\bm{p},\bm{q},\bm{r},a,b,\tj_1,\tj_2,\fii_1,\fii_2}(B)\,\d p_2(\bm{p},\bm{q},\bm{r},a,b,\tj_1,\tj_2,\fii_1,\fii_2)
$$
for all $B\in\mc L(\C^2)$.

\section{Conclusions}

We have given barycentric decompositions for instruments with a separable input Hilbert space, finite-dimensional output space, and a value set which is a locally compact second-countable Hausdorff space. As special cases, we obtain barycentric decompositions for quantum measurements in a separable Hilbert space and value sets like those above (generalizing the results of \cite{ChDASc}), for quantum channels with a separable (Schr\"{odinger}) input space and finite-dimensional output space, and for quantum states. In all these cases, the probability measure whose barycentre coincides with the quantum device of interest can be assumed to be supported by the set of extreme points. This means that we may write instruments, measurements, channels, and states satisfying the above (very loose) conditions as convex integrals over the sets of extreme instruments, measurements, channels, and pure states respectively. The result on states is, naturally, rather trivial since spectral decompositions already give barycentric decompositions for states over pure states.

Instead of the set of all instruments (or all measurements or channels) we are often interested in particular convex subsets of instruments. In such a situation, it is natural to look into the barycentric decompositions of instruments over the set of extreme points of this more specific convex set. One particularly interesting example of such subsets is covariance structures. Here, we have (strongly continuous) representations $U$ and $V$ of a (locally compact) group $G$ in $\hi$ (input space) and $\mc K$ (output space), respectively. We assume that $G$ continuously acts on the value set $\Omega$ which is also locally compact and second-countable Hausdorff space, i.e.,\ we have the continuous map $(G,\Omega)\ni(g,x)\mapsto gx\in X$ with $(gh)x=g(hx)$ and $ex=x$ for all $g,h\in G$ and $x\in\Omega$ where $e\in G$ is the neutral element. The $(U,V,\Omega)$-covariance structure is now the convex set of those instruments $\M\in\In\big(\mc B(\Omega),\mc K,\mc H)$ such that
$$
\M\big(gX,V(g)BV(g)^*\big)=U(g)\M(X,B)U(g)^*
$$
for all $g\in G$, $X\in\mc B(\Omega)$, and $B\in\lk$. Instruments like these reflect the physical symmetries of the quantum input and output systems and the value set with respect to the symmetry group $G$. The extreme points of such sets can be characterized concisely at least when the stability subgroup of $G$ with respect to the value set $\Omega$ is compact \cite{HaPe2021}. It is natural next to give $(U,V,\Omega)$-covariant instruments barycentric decompositions over the extreme points of the $(U,V,\Omega)$-covariance structure. However, we leave this question for future study.

\section*{Acknowledgements}

RU is thankful for the financial support from the Swiss National Science Foundation (Ambizione PZ00P2-202179).
EH is supported by the National Research Foundation, Singapore and A*STAR under its CQT Bridging Grant.

\end{document}